\documentclass[9pt, twocolumn]{extarticle}
\usepackage[margin=0.8in]{geometry}

\usepackage[T1]{fontenc}
\usepackage{utopia}

\usepackage{graphicx} 
\usepackage{appendix}

\usepackage{amsthm}
\usepackage{amsmath}
\usepackage{amsfonts}

\usepackage{setspace}

\theoremstyle{plain}
\newtheorem{theorem}{Theorem}[section]
\newtheorem{lemma}[theorem]{Lemma}

\theoremstyle{definition}

\theoremstyle{remark}

\newcommand\norm[1]{\left\lVert#1\right\rVert}
\usepackage{xcolor}

\newcommand\blfootnote[1]{
    \begingroup
    \renewcommand\thefootnote{}\footnote{#1}
    \addtocounter{footnote}{-1}
    \endgroup
}

\title{\huge{Online input design for discrimination of linear models using concave minimization}}
\author{Jacques Noom\textsuperscript{a}, Oleg Soloviev\textsuperscript{a,b}, Carlas Smith\textsuperscript{a} and Michel Verhaegen\textsuperscript{a}\\
\textsuperscript{a}Delft Center for Systems and Control, Delft University of Technology, \\2628CN Delft, The Netherlands; \\\textsuperscript{b}Flexible Optical BV, 2288GG Rijswijk, The Netherlands}
\date{}

\begin{document}

\maketitle

\begin{abstract}
    \textbf{Stochastic Closed-Loop Active Fault Diagnosis (CLAFD) aims to select the input sequentially in order to improve the discrimination of different models by minimizing the predicted error probability. As computation of these error probabilities encompasses the evaluation of multidimensional probability integrals, relaxation methods are of interest. This manuscript presents a new method that allows to make an improved trade-off between three factors -- namely maximized accuracy of diagnosis, minimized number of consecutive measurements to achieve that accuracy, and minimized computational effort per time step -- with respect to the state-of-the-art. It relies on minimizing an upper bound on the error probability, which is in the case of linear models with Gaussian noise proven to be concave in the most challenging discrimination conditions.
    A simulation study is conducted both for open-loop and feedback controlled candidate models. The results demonstrate the favorable trade-off using the new contributions in this manuscript.}
\end{abstract}
\blfootnote{This project has received funding from the ECSEL Joint Undertaking (JU) under grant agreement No 826589. The JU receives support from the European Union’s Horizon 2020 research and innovation programme and Netherlands, Belgium, Germany, France, Italy, Austria, Hungary, Romania, Sweden and Israel.}
\blfootnote{Compiled \today.}

\section{Introduction}
Fault diagnosis is crucial for automation. Widespread applications rise from recognizing faults in dynamical systems \cite{Gao2015} to automatic classification of images \cite{Rawat2017,Noom2022}. Moreover, Industry 4.0 aims at fully automated, smart factories where (among other things) ample data are turned into automatic actions and decisions. The motivation of such actions and decisions lies in model diagnosis. It is therefore critical that the diagnosis is efficient and reliable.

Passive approaches of fault diagnosis have the potential to overlook faults, as complex or feedback-controlled systems may generate nominal input-output data while not being fault-free \cite{Heirung2019}.
Active fault diagnosis overcomes the shortcomings of passive diagnosis by applying an input designed such that complex models can still be discriminated. Online computation of the discriminating inputs further increases the efficiency by considering most recent measurements. Such closed-loop implementations receive increased interest for their decreased conservatism and accelerated diagnosis, yet are not widespread use in automation \cite{Heirung2019}.

Although the idea of Closed-Loop Active Fault Diagnosis (CLAFD) exists for several decades in static \cite{Hunter1965,Box1967} and dynamical systems \cite{Zhang1988}, the developments were held back due to its computational challenges \cite{Heirung2019}. Moreover, the trade-off consists of three factors. First, the accuracy of diagnosis should be maximized. Second, the system should be diagnosed within a minimized number of consecutive measurements. Third, the inputs should be calculated fast enough so they can be applied to the system without delay.
This implication between performance (first two factors) and computational efficiency (third factor) is still a major bottleneck.

Existing efforts can be separated in deterministic and stochastic approaches \cite{Puncochar2018}, which are both computationally challenging. Assuming bounded uncertainties, a deterministic approach facilitates guaranteed diagnosis of the correct model (e.g. \cite{Yang2021}). This problem is computationally challenging due to the nonconvex constraints on the input such that the system output is exclusive \cite{Raimondo2016}. Alternatively, stochastic approaches assume known Probability Density Functions (PDFs) for unbounded uncertainties. Instead of anticipating the worst-case scenario for guaranteed diagnosis, the goal of stochastic approaches is to minimize the probability of misdiagnosis.
This is computationally challenging because of the (online) evaluation of multidimensional probability integrals to determine the error probability.
An illustration of the enormous computational challenge here is reported in \cite{Puncochar2015}, where even when restricting the input to a small discrete set of three elements and three second order candidate models, the computational time for determining the input policy was 7.5 hours \cite{Puncochar2015}.
Yet, stochastic approaches are generally less intrusive than deterministic approaches \cite{Scott2013}. Besides, it is in practice often difficult to define explicit bounds on signals, which makes it more natural to formulate CLAFD as a stochastic optimization problem \cite{Paulson2018,Noom2021}.

To overcome the tremendous computational burden of stochastic approaches, a widely accepted solution is to optimize an upper bound, such as the sum of weighted Bhattacharyya coefficients \cite{Boekee1979}. Using sequential quadratic programming and restricting to open-loop (batch-wise) input determination, this has been studied in \cite{Blackmore2006}. Due to the remaining high computational complexity, closed-loop (receding-horizon) implementation is still unattractive. \cite{Noom2021} restricted to polytopic constraints on the input in an attempt to implement it efficiently in closed-loop. This approach however could not guarantee a solution that is optimal with respect to the upper bound. 

As an alternative for overcoming the real-time bottleneck, a further simplification is proposed in \cite{Paulson2018}, where just the convex part of the Bhattacharyya coefficient is employed. The so-called Bhattacharyya \emph{distance} was maximized in an attempt to discriminate the candidate models. Although this resulted in a fast computational solution, serious drawbacks are introduced. First, the function used in \cite{Paulson2018} is by no means an approximation to the upper bound on the error probability. Consequently, there exist conditions where the performance degrades substantially, such as the case of multiple groups of close candidate models. Second, \cite{Paulson2018} only considered polytopic constraints on the input in order to determine the maximum of the simplified function.

The main contribution of this manuscript is as follows. First, for linear systems with Gaussian disturbances it is shown that the upper bound on the error probability is concave in a subdomain of the input space. That subdomain is characterized by the interesting case when discrimination between candidate models is challenging in terms of noise and proximity of the candidate models. Second, a real-time online check for being in this subdomain is formulated both for the case of a polytopic and quadratic constraint set. Third, we propose to solve the minimization problem using Disciplined Convex-Concave Programming (DCCP) \cite{Shen2016}, extending the ability to implement the closed-loop diagnosis procedure to a broad spectrum of convex constraints on the input. Altogether, this overcomes the two main drawbacks of the approach in \cite{Paulson2018}. For further speedup of the computations without unnecessarily sacrificing the performance, a quadratic Taylor approximation of the upper bound on the error probability is proposed. The differences of the proposed approaches with that in \cite{Paulson2018} are verified in a first Monte-Carlo simulation, after which a second Monte-Carlo simulation aims at comparing the performances of the approaches extended by DCCP. The final simulation shows the applicability of the proposed approaches to feedback controlled systems.

The manuscript is organized as follows. First, the stochastic diagnosis problem is defined in Sect. \ref{probform}. Then, the practical implementation is presented in Sect. \ref{methodology}, together with the derivation of domain of concavity and the quadratic Taylor approximation. Sect. \ref{simulation} presents simulation results for the proposed closed-loop methods, along with the state-of-the-art closed-loop method based on Bhattacharyya distances and an open-loop approach. Lastly, Sect. \ref{conclusion} states the conclusions.

\section{Problem Formulation} \label{probform}
\subsection{Notation}
In consistency with \cite{Paulson2018,Noom2021}, we will use the following notation. The expression $x|y$ denotes the random variable $x$ conditioned on $y$, $\hat{x}_{k|k-n}$ is the estimate of $x_k$ based on knowledge up until time step $k-n$, the notation $x_{k:k+n} = \begin{bmatrix} x_k^\top & x_{k+1}^\top & \hdots & x_{k+n}^\top \end{bmatrix}^\top$, and we define boldface $\mathbf{x}_k = x_{k+1:k+N}$ with $N$ the horizon length. The notation $P_k(\text{event}) = \mathbb{P}(\text{event at time step }k)$ with $\mathbb{P}$ the probability operator, and the expectation operator $\mathbb{E}[\cdot]$ works on the stochastic output data $y_k$.

\subsection{Bayesian formulation}
Regard the linear candidate models of a system:
\begin{equation}
M^{[i]}: \left\{ \begin{array}{rl} x_{k+1} &= A^{[i]} x_{k} + B^{[i]} u_k + w_k\\
y_k &= C^{[i]} x_{k}  +  v_k \end{array}\right. 
\label{models}
\end{equation}
with $i = \{0,1,\hdots,n_m-1\}$ the model indicator, $A^{[i]},B^{[i]},C^{[i]}$ the state space matrices and $x_k \in \mathbb{R}^{n_x}$, $u_k \in \mathbb{R}^{n_u}$, $y_k\in \mathbb{R}^{n_y}$ the state, input and output. Given the joint covariance matrix of the process noise $w_k$ and measurement noise $v_k$
\begin{equation}
    \mathbb{E}\left[\begin{bmatrix}
        v_k \\ w_k
    \end{bmatrix} \begin{bmatrix}
        v_\ell^\top & w_\ell^\top
    \end{bmatrix}\right] = 
    \left\{ \begin{array}{ll} \begin{bmatrix}
        R & S^\top \\ S & Q
    \end{bmatrix} & \text{if } k = \ell, \\
    0 & \text{otherwise,}
    \end{array} \right.
    \label{modelsnoise}
\end{equation}
the optimal state prediction $\hat{x}_{k+1|k}^{[i]}$ can be obtained together with its corresponding covariance matrix $\Xi_{k+1|k}^{[i]}$ using the Kalman filter as described in for instance \cite{Verhaegen2007}. Stability of the candidate models is not required. Nevertheless, feedback controlled systems can be captured in \eqref{models} and \eqref{modelsnoise} as described in Appendix \ref{FBC}.

The hypothesis probabilities evolve according to the Bayesian update rule
\begin{equation}
    \begin{aligned}
    &P_{k+N}\left(M^{[i]}\big|P_k(M^{[i]}),\hat{x}_{k+1|k}^{[0:n_m-1]},\mathbf{y}_{k},\mathbf{u}_k\right) \\ &\qquad\qquad =\frac{p\left(\mathbf{y}_{k}\big|M^{[i]},\hat{x}_{k+1|k}^{[i]},\mathbf{u}_{k}\right) P_{k}(M^{[i]})}{p\left(\mathbf{y}_{k}\big|\hat{x}_{k+1|k}^{[0:n_m-1]},\mathbf{u}_{k}\right)},
    \end{aligned}
    \label{probs}
\end{equation}
where $p\left(\mathbf{y}_{k}\big|M^{[i]},\hat{x}_{k+1|k}^{[i]},\mathbf{u}_{k}\right) \in \mathbb{R}_+$ is the probability density function (PDF) of output $\mathbf{y}_{k}$, conditioned on hypothesis $M^{[i]}$, state estimate $\hat{x}_{k+1|k}^{[i]}$ and input $\mathbf{u}_{k}$. The initial conditions $P_{0}(M^{[i]})$ can be set to any prior probabilities. Observe that in this notation the left-hand side of \eqref{probs} is conditioned on another probability, which is convenient for the reason that the model probabilities $P_k(M^{[0:n_m-1]})$ combined with system state estimates $\hat{x}_{k+1|k}^{[0:n_m-1]}$ fully describe the so-called belief state of the partially observable Markov decision process.

If one chooses the most likely model
 \begin{equation}
     M^{[i]} \; : \; i = \arg \max_i \Big(P_{k}(M^{[i]}) \Big)
 \end{equation}
 based on all knowledge up until time step $k$, then the probability of misdiagnosis would be
\begin{equation}
P_{k}\Big(\text{error}\big|P_{k}(M^{[0:n_m-1]}) \Big) =  1 - \max_i \Big(P_{k}(M^{[i]}) \Big).
\label{perror1}
\end{equation}
The \emph{future} (unknown) probability of misdiagnosis after next measurements $\mathbf{y}_{k}$ also depends on future inputs and outputs:
\begin{equation}
\begin{aligned}
&P_{k+N}\left(\text{error}\big|P_{k+N}(M^{[0:n_m-1]}) \right)\\
&\quad = P_{k+N}\left(\text{error}\big|P_{k}(M^{[0:n_m-1]}),\hat{x}_{k+1|k}^{[0:n_m-1]},\mathbf{y}_{k},\mathbf{u}_k \right)\\
&\quad =  1 - \max_i \left(P_{k+N}\left(M^{[i]}\big|P_k(M^{[i]}),\hat{x}_{k+1|k}^{[0:n_m-1]},\mathbf{y}_{k},\mathbf{u}_k\right) \right).
\end{aligned}
\label{perror2}
\end{equation}
Obviously, it is desirable to minimize this future error probability. For that purpose, the stochastic problem of CLAFD is based on the expected value of the error probability. Besides, one could imagine that the optimal input sequence for minimizing the error probability changes when more measurements become available. This is covered in the general stochastic formulation of the CLAFD problem, as presented in next section.

\subsection{Infinite horizon stochastic control problem}
In accordance with \cite{Bertsekas1978}, the stochastic control problem of CLAFD is generally formulated as
\begin{equation}
\begin{aligned}
\min_\pi\; &\lim_{N \rightarrow \infty} \sum_{k=0}^{N-1} \mathbb{E} \left[P_{k}\Big(\text{error}\big|s_k \Big) \Big| s_0\right]\\
\text{s.t. } &\pi \in \Pi\\
&s_{k+1} = f[s_k,\mu_k(s_k),\mathrm{w}_k].
\end{aligned}
\label{stoc}
\end{equation}
Here, the input policy $\pi = (\mu_0,\hdots,\mu_{N-1})$ defines the input with the functions
\[
u_k = \mu_k(s_k),
\]
constrained to the set $\Pi$. The function $f[s_k,\mu_k(s_k),\mathrm{w}_k]$ describes the system dynamics as function of state $s_k$, input function $\mu_k(s_k)$ and $\mathrm{w}_k$ the stochastic disturbance. Similarly to \cite{Puncochar2015}, the hyperstate $s_k$ should consist of the state estimates for each candidate model $\hat{x}_{k+1|k}^{[i]}$ and the model probabilities $P_k(M^{[i]})$. The variable $\mathrm{w}_k$ consists of the noise contributions $w_k$ and $v_k$.

The expected value of the error probability in \eqref{stoc} is also known as the \emph{risk} in a Bayes classifier (see e.g. \cite{Rish2001}). When considering for the case in this manuscript that the estimated states and covariance matrices are available from \eqref{models}, the expected value acts as a predictor of the error probability in \eqref{perror2}. In other words,
\begin{equation}
    \begin{aligned}
    &\hat{P}_{k+N|k}\left(\text{error}\big|\mathbf{u}_k\right) \\
    &\quad =\mathbb{E}\left[P_{k+N}\left(\text{error}\big|P_{k}(M^{[0:n_m-1]}),\hat{x}_{k+1|k}^{[0:n_m-1]},\mathbf{y}_{k},\mathbf{u}_k \right) \right]\\
    &\quad = P_{k+N}\left(\text{error}\big|P_{k}(M^{[0:n_m-1]}),\hat{x}_{k+1|k}^{[0:n_m-1]},\mathbf{u}_k \right).
    \end{aligned}
\end{equation}
Note that all sides of this equation are independent of the future system outputs $\mathbf{y}_{k}$, and therefore it is a deterministic function of $\mathbf{u}_k$ only.

\section{Methodology} \label{methodology}
\subsection{Model predictive control}
A generally accepted strategy for solving an infinite horizon stochastic control problem is to make use of a receding horizon approach, most commonly known as Model Predictive Control (MPC) \cite{Mayne2000}. This technique solves a finite horizon control problem each time step, and applies only the first input to the system. The $N$-step receding-horizon approximation of \eqref{stoc} is
\begin{equation}
\begin{aligned}
\mathbf{u}_k^* = \arg \min_{\mathbf{u}_k \in \mathcal{U}} &\sum_{\ell = k+2}^{k+N} \hat{P}_{\ell|k}\left(\text{error}\big|u_{k+1:\ell-1}\right)
\end{aligned}
\label{optprob1}
\end{equation}
with $\mathcal{U}$ the constraint set on vectorised inputs $\mathbf{u}_k$.

\subsection{Bound on predicted error probability}
Since the prediction of the error probability requires large computational effort, a divergence measure will be used instead. The Bhattacharyya coefficient is a symmetric measure which provides an upper bound on the predicted error probability \cite{Boekee1979}:
\begin{equation}
\hat{P}_{k+N|k}\left(\text{error}\big|\mathbf{u}_k\right) \leq \sum_i \sum_{j>i} \sqrt{P_{k}(M^{[i]}) P_{k}(M^{[j]})} \mathfrak{B}^{[ij]}(\mathbf{u}_k)
\label{bound}
\end{equation}
with the Bhattacharyya coefficient in case of Gaussian process and measurement noise
\begin{equation}
\mathfrak{B}^{[ij]}(\mathbf{u}_k) = \exp(-d^{[ij]}(\mathbf{u}_k)).
\end{equation}
The Bhattacharyya distance 
\begin{equation}
d^{[ij]}(\mathbf{u}_k) =  \mathbf{u}_k^\top H^{[ij]} \mathbf{u}_k +  \mathbf{u}_k^\top c^{[ij]} + h^{[ij]}
\label{bdist}
\end{equation}
 is a convex quadratic function where \cite{Paulson2018}:
 \begin{equation}
     \begin{aligned}
     H^{[ij]} &= \frac{1}{4}(\Gamma^{[ij]})^\top (\Omega^{[ij]})^{-1} \Gamma^{[ij]}\\
     c^{[ij]} &= \frac{1}{2} (\Gamma^{[ij]})^\top (\Omega^{[ij]})^{-1} \zeta^{[ij]}
     \end{aligned}
     \label{Hc}
 \end{equation}
 and
  \begin{equation}
     \begin{aligned}
     \Omega^{[ij]} &= \mathbf{\Sigma}_{k|k}^{[i]} + \mathbf{\Sigma}_{k|k}^{[j]}\\
     \Gamma^{[ij]} &= \mathbf{C}^{[i]}\mathcal{T}_{A^{[i]}} \mathbf{B}^{[i]} - \mathbf{C}^{[j]}\mathcal{T}_{A^{[j]}}\mathbf{B}^{[j]} \\
     \zeta^{[ij]} &= \mathbf{C}^{[i]}\mathbf{A}^{[i]} \hat{x}_{k+1|k}^{[i]} - \mathbf{C}^{[j]}\mathbf{A}^{[j]} \hat{x}_{k+1|k}^{[j]}
     \end{aligned}
     \label{OmGamZet}
 \end{equation}
 with the matrices $\mathbf{A}^{[i]},\mathbf{B}^{[i]},\mathbf{C}^{[i]}$ and Toeplitz matrix $\mathcal{T}_{A^{[i]}}$ constructed from the state-space matrices in \eqref{models}, and $\mathbf{\Sigma}_{k|k}^{[i]}$ are the covariance matrices of the estimates of the output $\hat{\mathbf{y}}_{k|k}^{[i]}$. The full definitions together with that of $h^{[ij]}$ are given in Appendix \ref{definitions}.

Given the bound in \eqref{bound}, the relaxed MPC problem reads
\begin{equation}
\begin{aligned}
\mathbf{u}_k^* = \arg \hspace{-0.04in} \min_{\mathbf{u}_k \in \mathcal{U}} \sum_{\ell = k+2}^{k+N}\hspace{-0.05in}\sum_i \hspace{-0.02in}\sum_{j>i} \hspace{-0.03in}\sqrt{P_{k}(M^{[i]}) P_{k}(M^{[j]})} \mathfrak{B}^{[ij]}(u_{k+1:\ell-1}).
\end{aligned}
\label{optprob}
\end{equation}
Since the objective function \eqref{optprob} is still non-convex, efforts have been made to simplify it further. The approach in \cite{Noom2021} aims to solve 
\begin{equation}
\begin{aligned}
\mathbf{u}_k^* = \arg \min_{\mathbf{u}_k \in \mathcal{U}} \sum_i \sum_{j>i} \sqrt{P_{k}(M^{[i]}) P_{k}(M^{[j]})} \mathfrak{B}^{[ij]}(\mathbf{u}_{k})
\end{aligned}
\label{optprob2}
\end{equation}
instead, using only the error bound for time step $k+N$. Furthermore,  \cite{Paulson2018} simplifies problem \eqref{optprob2} to 
\begin{equation}
\begin{aligned}
\mathbf{u}_k^* = \arg \min_{\mathbf{u}_k \in \mathcal{U}} \sum_i \sum_{j>i} - d^{[ij]}(\mathbf{u}_k)
\end{aligned}
\label{paulson}
\end{equation}
in which only the Bhattacharyya distances $d^{[ij]}(\mathbf{u}_k)$ are considered. Its performance is demonstrated on a problem with a polytopic constraint set, such that the minimum of this concave function lies at one of the vertices of this set. Although the computational effort is relatively small, the solution can differ significantly from the optimum of \eqref{optprob2}. To improve the optimization outcome, \cite{Noom2021} proposes to evaluate the objective function in \eqref{optprob2} only at the vertices of the constraint set. Still, the global optimum of \eqref{optprob2} is not guaranteed using this approach.

The current manuscript follows the observation that \eqref{optprob} is in closed-loop diagnosis -- with challenging discrimination conditions -- often concave within the given constraint set. The next sections elaborate on verification of concavity and show how this concave problem can be minimized for polytopic constraints and energy constraints. Besides, Sect. \ref{taylorsect} presents an improved quadratic approximation w.r.t. \eqref{paulson}, which can be used if regular concave minimization is still too demanding with respect to quadratic concave minimization.

\subsection{Domain of concavity of the Bhattacharyya coefficient}\label{bhatdomain}
Note that in case of Gaussian process and measurement noise, the Bhattacharyya coefficient is a Gaussian function
\begin{equation}
\mathfrak{B}(\mathbf{u}_k) = \exp \left( - \mathbf{u}_k^\top H \mathbf{u}_k - c^\top \mathbf{u}_k - h \right)
\label{firstb}
\end{equation}
with $H$ a positive semi-definite, symmetric matrix (the indices $ij$ are omitted for clarity) which can be partitioned using the singular value decomposition (SVD)
\begin{equation}
H = \begin{bmatrix}
U_1 & U_2
\end{bmatrix}
\begin{bmatrix}
\Lambda_1 & 0 \\ 0 & 0
\end{bmatrix}
\begin{bmatrix}
U_1^\top \\ U_2^\top
\end{bmatrix}.
\end{equation}
The domain where the Bhattacharyya coefficient is concave, is provided in the following lemma.

\begin{lemma}[Domain of concavity of a multivariate Gaussian]\label{domconcavity}
Expression \eqref{firstb} is concave in the domain where
\begin{equation}
 \mathbf{u}_k^\top H \mathbf{u}_k + c^\top \mathbf{u}_k + \frac{1}{4}c^\top U_1 \Lambda_1^{-1} U_1^\top c \leq \frac{1}{2}
 \label{conc}
\end{equation}
is satisfied.
\end{lemma}
\begin{proof}
The Bhattacharyya coefficient can be represented as
\begin{equation}
\mathfrak{B}(\mathbf{u}_k) = a \exp(-\rho^2)
\label{secondb}
\end{equation}
with 
\begin{equation}
\begin{aligned}
\rho^2 &= g^\top g\\
g &= \Lambda_1^{\frac{1}{2}} U_1^\top \mathbf{u}_k + \frac{1}{2} \Lambda_1^{-\frac{1}{2}} U_1^\top c\\
a &= \exp \left( \frac{1}{4} c^\top U_1 \Lambda_1^{-1} U_1^\top c - h \right).
\end{aligned}
\label{defs}
\end{equation}
The second derivative of \eqref{secondb} with respect to the radius $\rho$ is non-positive for $\rho^2 \leq \frac{1}{2}$, implying that the concave part of \eqref{firstb} is the domain where \eqref{conc} is satisfied.
\end{proof}
As illustrated in Fig. \ref{concavity_plot}, Lemma \ref{domconcavity} implies that concavity is likely in the challenging conditions when the model differences indicated by $\Gamma^{[01]}$ in \eqref{OmGamZet} are small and the noise contribution indicated by $R$ is large. The term $c$ grows with differences in state estimates per model. As long as the input remains small enough, and the candidate models do not differ significantly in resonance frequencies, the term $c$ will remain small and the optimization problem is likely to be completely in the domain of concavity. Note that these conditions are again the most challenging for discrimination of models, whereas minimization of the Bhattacharyya coefficient is a concave problem for these conditions, which is computationally tractable for multiple types of constraints.

\begin{figure}[ht]
\centering
\includegraphics[width=0.9\columnwidth]{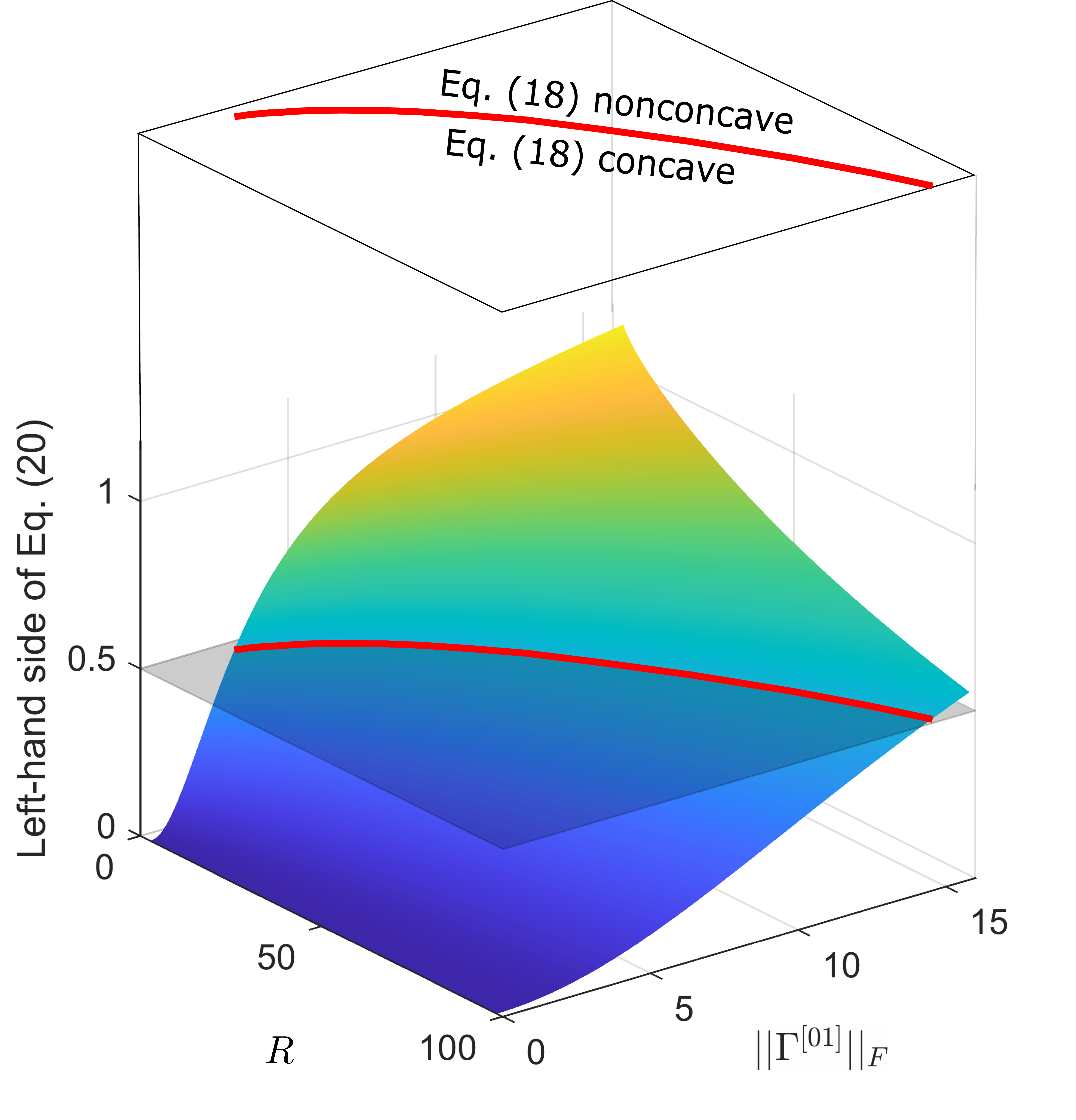}
\caption{The compliance of \eqref{conc} against measurement noise variance $\mathtt{R}$ and model differences $\mathtt{\norm{\Gamma^{[01]}}_F}$, generated with initial conditions, horizon $N$ and model $\mathtt{M^{[0]}}$ as in Sect. \ref{clamduncontrolled}, $\mathtt{C^{[1]}} = \text{constant} \cdot \mathtt{C^{[0]}}$, and with inputs $\mathbf{u}_k =  \mathtt{[-1,\,-1,\,-1,\,-1,\,1,\,1,\,-1,\,-1,\,0,\,0]^\top}$. Large noise contribution and small model differences are favorable for satisfying \eqref{conc}.}
\label{concavity_plot}
\end{figure}

\subsection{Online check for concavity of \eqref{optprob}} \label{onlinecheck}
\subsubsection{For polytopic constraints}\label{checkpoly}
To see whether \eqref{optprob} is concave in the case of polytopic constraints, it is sufficient to check whether all vertices of the constraint set satisfy \eqref{conc}, for each model combination $ij$.

\subsubsection{For energy constraints}\label{checkenergy}
 The input energy constraint set has the form
\begin{equation}
    \mathcal{U} = \left\{\mathbf{u}_k \;\Big|\; \norm{u_\ell}_2^2 \leq \varepsilon , \quad k+1 \leq \ell \leq k+N \right\}.
    \label{energyconstr}
\end{equation}
If the input sequence $\mathbf{u}_k = 0$ satisfies \eqref{conc}, it can be deduced that any point $\mathbf{u}_k$ for which $\norm{\mathbf{u}_k}_2 \leq \norm{z^*}_2$, with
\begin{equation}
\begin{aligned}
z^* = \arg &\min_z z^\top z \\
&\text{s.t. } z^\top H z + c^\top z + \frac{1}{4} c^\top U_1 \Lambda_1^{-1} U_1^\top c = \frac{1}{2}
\end{aligned}
\label{min}
\end{equation}
is within the concave domain of \eqref{firstb}. This optimization can be solved using only small computational effort as is shown in Appendix \ref{concavitysol}.
Since there is an input energy constraint for each time step in the horizon, 
$$\norm{u_{k+1}}_2^2 + \norm{u_{k+2}}_2^2 + \hdots + \norm{u_{k+N}}_2^2 \leq  \varepsilon N$$
holds and therefore the optimization problem will be in the domain of concavity of \eqref{optprob} if $\sqrt{\varepsilon N} \leq \norm{z^*}_2$ is satisfied
for each model combination $ij$.

These tests for concavity are performed each time step before the input determination. The assessment of concavity prior to the full closed-loop discrimination experiment requires closer investigation in future research. This deficiency however does not prevent the closed-loop approach from diagnosing the true candidate model. Instead, it is only uncertain whether the solution found each time step is the actual optimum of \eqref{optprob}.

\subsection{Quadratic Taylor approximation of \eqref{optprob}} \label{taylorsect}
The quadratic Taylor approximation of the Bhattacharyya coefficient around $\mathbf{u}_k = 0$ is
\begin{equation}
\begin{aligned}
T_\mathfrak{B}(\mathbf{u}_k | 0) = \mathfrak{B}(0) \left( \frac{1}{2} \mathbf{u}_k^\top \left(c c^\top - 2H \right) \mathbf{u}_k - c^\top \mathbf{u}_k+1\right).
\end{aligned}
\label{taylor}
\end{equation}
The Taylor approximation of the sum of weighted coefficients, such as the right-hand side of \eqref{bound}, can be obtained by taking the weighted sum of individual Taylor approximations. This means that minimization \eqref{optprob} is approximated by
\begin{equation}
\begin{aligned}
\mathbf{u}_k^* = \arg \min_{\mathbf{u}_k \in \mathcal{U}} \sum_{\ell = k+1}^{k+N} \sum_i \sum_{j>i} \sqrt{P_{k}(M^{[i]}) P_{k}(M^{[j]})} T_\mathfrak{B}(u_{k+1:\ell} | 0).
\end{aligned}
\label{taylorapprox}
\end{equation}
 Note that this quadratic approximation is significantly different from simply taking the sum of Bhattacharyya distances as in \eqref{paulson}. It is more reasonable, since the error bound in \eqref{bound} is now appropriately approximated, rather than only disregarding the model probabilities and base of exponentials as is done in \eqref{paulson}.
 
 \subsection{Summary of approaches}\label{summary}
Four MPC-based approximations and one open-loop (OL) approach of \eqref{stoc} will be compared. The Bhattacharyya Distance (BD) approach is reproduced from \cite{Paulson2018} and solves \eqref{paulson} in a receding horizon manner. The Bhattacharyya Coefficient (BC) approach solves \eqref{optprob2} instead, given that the problem is within the domain of concavity. The validity of the optimum can be checked as described in Sect. \ref{onlinecheck}. The Summed Bhattacharyya Coefficient ($\Sigma$BC) takes multiple time steps into account by solving \eqref{optprob}, again while considering the domain of concavity. Next, the Quadratic Taylor Approximation (QTA) approach solves \eqref{taylorapprox}, yet only considering the error bound at time step $k+N$ in order to make a fair comparison with BD and BC. Lastly, the OL approach minimizes \eqref{optprob2} offline (i.e., only at time step $k=0$) with very large horizon $N$.

\section{Simulation Experiment}\label{simulation}
The proposed approaches are first tested in a setting with a polytopic constraint set on the inputs, as is done in \cite{Paulson2018}. Moreover, a Monte-Carlo simulation is conducted in order to assess the performance more generally. Additionally, they will be tested in a case with quadratic input constraints, aided by the DCCP toolbox \cite{Shen2016}. A final simulation in Sect. \ref{clamdcontrolled} demonstrates how the methodology can be applied to a feedback controlled system.

The approaches were tested on the simulation setup given in \cite{Paulson2018}. The performances of the approaches were very similar to each other due to the chosen set of candidate models and small noise contribution, such that a well-designed input was not crucial for discrimination. In this section, we therefore created a more challenging case study, namely with multiple groups of close candidate models and large noise contribution.

\subsection{Closed-loop diagnosis of an uncontrolled system}\label{clamduncontrolled}
The candidate models in the simulation are constructed from a continuous-time harmonic oscillator with resonance $\pi/2$, damping $0.1$ and discretized with sampling time $1$, such that the matrices\footnote{Note that these values are rounded. Simulation results using rounded values may differ from the results presented in this manuscript. Besides, note that extension of the presented methodology to time-variant state-space matrices and model-dependent noise characteristics is straightforward.}
\begin{equation}
\begin{aligned}
A^{[i]} &=  \begin{bmatrix}
-0.0792 + \Delta^{[i]} & -0.6746 \\ 1.0936 & 0.0926
\end{bmatrix}, \\
B^{[i]} &= \begin{bmatrix}
0.2734 & 1.5700 - \delta^{[i]} \\ 0.3677 & 0
\end{bmatrix}\\
\Delta &= \{0,\,0.2,\,0.4,\,1.0,\,1.1\}\\
\delta &= \{0,\,0.1660,\,0.3319,\,0.8297,\,0.9127\}\\
\end{aligned}
\label{modelchoice1}
\end{equation}
variable per model, and
\begin{equation}
\begin{aligned}
C^{[i]} = C = \begin{bmatrix}
0 & 1 \\ 0.1 & 0.5
\end{bmatrix}
\end{aligned}
\label{modelchoice2}
\end{equation}
equivalent for each model, with in total five models $M^{[i]}$ constructed from taking $i=\{0,1,2,3,4\}$. The Gaussian noise contribution is relatively high with $Q=0.2 I_2$, $R=80 I_2$ and $S=0$. Given these quantities, a Kalman filter \cite{Verhaegen2007} is constructed for online estimation of the system state and output, which are normally distributed. The initial state is $\hat{x}_{0|-1} = \begin{bmatrix}
0 & 1
\end{bmatrix}^\top$ with covariance $\Xi_{0|-1}= 0.5 I_2$, and the initial probabilities are $P_{0}(M^{[i]}) = 0.2$ for each candidate model. For online determination of a separating input, the BD, QTA, BC and $\Sigma$BC approaches are implemented with receding horizon of length 5. Either if a model probability exceeds $1 - \epsilon = 0.98$, or if the number of measurements exceeds its maximum of 400, the discrimination experiment stops and the algorithm makes a decision about the underlying model of the system.

The OL approach uses a horizon of 200 time steps. The optimum is found as in \cite{Blackmore2006} using sequential quadratic programming with 20 initializations. If a decision cannot be made after 200 time steps, the open-loop input is repeated.

For each experiment setting, the number of Monte-Carlo (MC) runs is 1000, with 200 per true candidate model, performed on an Intel i7-9750H CPU.

\subsubsection{Parameters for polytopic constraint set}
For the experiment with polytopic constraint set, the input is restricted to
\begin{equation}
\begin{aligned}
\mathcal{U} = \Big\{\mathbf{u}_k \;\Big|\; &\norm{u_\ell}_\infty \leq 2, \quad\\ &\norm{u_\ell - u_{\ell-1}}_\infty \leq 1 , \quad  k+1 \leq \ell \leq k+N \Big\}
\end{aligned}
\nonumber
\end{equation}
After a check of concavity as described in Sect. \ref{checkpoly}, the optima for problems \eqref{optprob}, \eqref{optprob2}, \eqref{paulson} and \eqref{taylorapprox} (for the $\Sigma$BC, BC, BD and QTA approach, respectively) are found using an exhaustive search over the vertices.

\subsubsection{Parameters for quadratic constraint set}
The experiment with quadratic constraint set has restriction \eqref{energyconstr} with $\varepsilon = 2$.
With a positive check for concavity as explained in Sect. \ref{checkenergy}, the minimization problems \eqref{optprob}, \eqref{optprob2}, \eqref{paulson} and \eqref{taylorapprox} are solved using DCCP \cite{Shen2016}.

\subsubsection{Results}
The results for the polytopic and quadratic constraint set are presented as violin plots \cite{Hintze1998,Bechtold2016} in Figs. \ref{violin_polytopic} and \ref{violin_quad1}, respectively. According to the online check as described in Sect. \ref{onlinecheck}, all computations were done in the domain of concavity of \eqref{optprob} for BC and $\Sigma$BC in the experiment with polytopic constraint set. With the quadratic constraint set, at least $59.9$\% of the measurements yielded a domain where \eqref{optprob2} (or \eqref{optprob} for $\Sigma$BC) is concave. In the other cases, a discriminating input will still be found, but without the guarantees of DCCP. As explained in Sect. \ref{bhatdomain}, in these cases it is less crucial for the input to be optimally discriminating.

\begin{figure}[t]
\centering
\includegraphics[width=1\columnwidth]{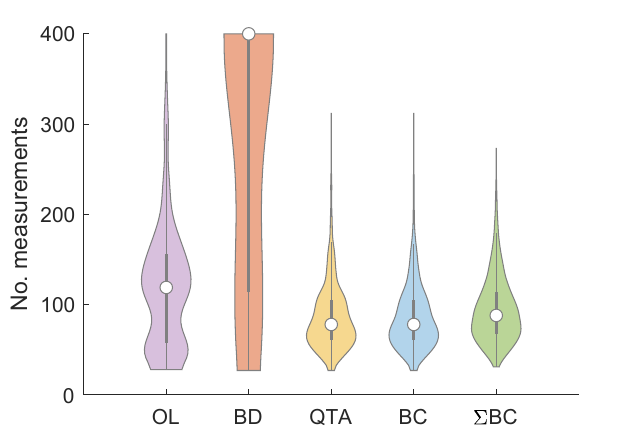}
\caption{Number of measurements before decision in experiment with polytopic constraints for the four closed-loop methods, compared to open-loop. The medians are from left to right $\mathtt{\{119,\, 400,\, 78,\, 78,\, 88\}}$ and the distributions of all methods differ significantly with MWW  approximated $p$-value $\mathtt{< 0.001}$, except the distribution pair $(\text{QTA},\text{BC})$. The average computational time per measurement for the closed-loop methods was $\mathtt{\{10.8,\, 12.4,\, 26.1,\, 89.1\}}$ milliseconds for BD, QTA, BC and $\Sigma$BC, respectively.}
\label{violin_polytopic}
\end{figure}
\begin{figure}[t]
\centering
\includegraphics[width=1\columnwidth]{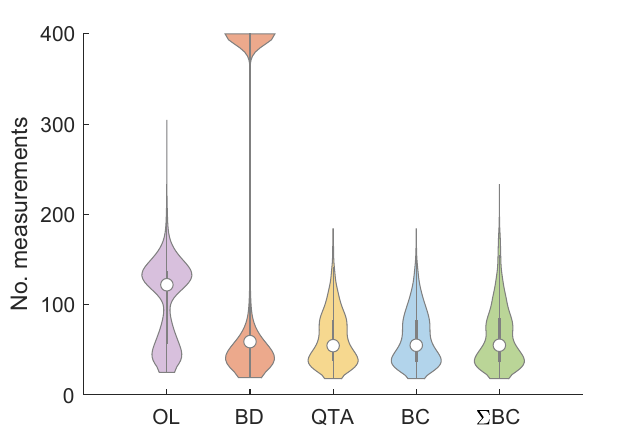}
\caption{Number of measurements before decision in experiment with quadratic constraints for the four closed-loop methods, compared to open-loop. The medians are from left to right $\mathtt{\{122,\, 59,\, 54.5,\, 55,\, 55\}}$. Although the distributions of QTA, BC and $\Sigma$BC do not differ significantly from each other, they do differ significantly from OL and BD with MWW approximated $p$-value $\mathtt{< 0.001}$. The average computational time per measurement for the closed-loop methods was $\mathtt{\{223,\, 175,\, 310,\, 928\}}$ milliseconds for BD, QTA, BC and $\Sigma$BC, respectively.}
\label{violin_quad1}
\end{figure}

\begin{figure}[ht]
\centering
\includegraphics[width=0.99\columnwidth]{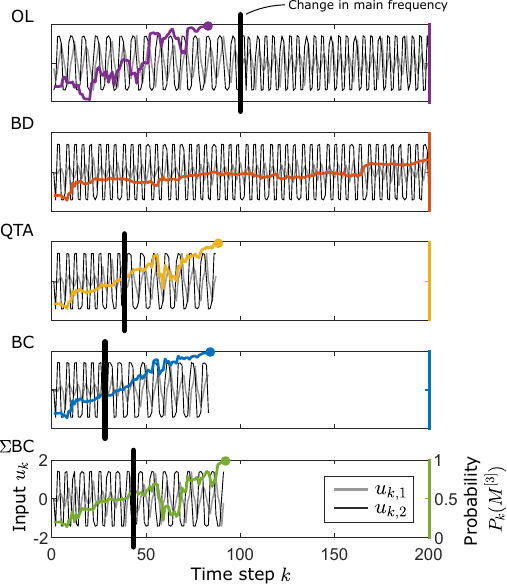}
\caption{Applied inputs $u_k = [u_{k,1}, u_{k,2}]^\top$ for one realization of the approaches summarized in Sect. \ref{summary} in quadratic constraint set, with corresponding probability of true model $\mathtt{M^{[3]}}$. The black vertical bars indicate time instances at which the input pattern changes significantly in main frequency. The colored circles indicate time instances when the final decision is made. The BD approach did not present the change in main frequency and failed to decide after 200 time steps.}
\label{realization3}
\end{figure}

Figs. \ref{violin_polytopic} and \ref{violin_quad1} show that in both cases with polytopic or quadratic constraints, the QTA, BC and $\Sigma$BC approaches require fewer measurements for diagnosis with confidence $1-\epsilon$ than OL and BD. 
The OL approach shows in both experiments a bimodal distribution for the required number of measurements. This is caused by the composition of candidate models, of which three models $M^{[0]}$, $M^{[1]}$ and $M^{[2]}$ have high resonance frequency, whereas the other two models $M^{[3]}$ and $M^{[4]}$ have relatively low resonance frequency. This implies that the input frequency should be either low or high to separate these respective models. Since the OL approach has a predetermined input sequence for discriminating all candidate models within a certain time span, it will invariably first start with one frequency for a fixed amount of time, after which it applies the other frequency for another fixed period. Therefore, depending on the resonance frequency of the true system, the diagnosis will be either early or late.

Fig. \ref{realization3} presents one realization of the closed-loop diagnosis approaches in the quadratic constraint set. The upper plot illustrates the latter phenomenon for the OL approach. As the black vertical bar indicates, the input frequency only increases after about 100 time steps, no matter what the system output is. In this realization, the true candidate model is a low-frequency model. Therefore the OL approach diagnoses it relatively quickly. Contrarily, the diagnosis in this experiment will be slower if the true model has high resonance frequency.

As opposed to open-loop, the closed-loop approaches do consider online measurements for updating the system inputs. Therefore, the QTA, BC and $\Sigma$BC approaches change the main input frequency as soon as the model probabilities `suggest' this, as can be seen in the bottom plots in Fig. \ref{realization3}. Although the distributions between QTA, BC and $\Sigma$BC do not differ significantly in Fig. \ref{violin_quad1}, the bottom plot in Fig. \ref{realization3} shows that individual realizations of these approaches can actually differ.

Interestingly, in many cases the BD approach fails to diagnose with $1-\epsilon$ confidence, even after 400 time steps. This is mainly due to disregarding the model probabilities in the optimization problem. As a result, it keeps trying to separate all five models simultaneously, while at some point several candidate models may become irrelevant due to low probability of being the true model. Another reason is that, even with equal model probabilities, optimization problem \eqref{paulson} is by no means an approximation to minimization of the original error bound \eqref{optprob2} for the fact that minimizing the sum of exponents produces a different result than minimizing the sum of exponentials.

In terms of computational time, BD, QTA and BC are comparable. However, there seems to be a preference to BD and QTA, which have quadratic objective functions. The $\Sigma$BC approach generally uses more computational time, which is reasonable since it has more terms in the summation in the objective function.

\subsection{Closed-loop diagnosis of a feedback controlled system} \label{clamdcontrolled}
The closed-loop procedure for diagnosing faults can also be applied to feedback controlled systems as illustrated in Fig. \ref{scheme} and described in Appendix \ref{FBC}. For this simulation experiment, the structure of the open-loop candidate models $\tilde{A}^{[i]},\tilde{B}^{[i]},\tilde{C}^{[i]}$ are chosen similar to \eqref{modelchoice1} and \eqref{modelchoice2} with
\begin{equation}
\begin{aligned}
    \Delta &= \{2,\,2.01,\,2.02,\,2.03,\,2.04\}\\
    \delta &= \{1.6594,\,1.6677,\,1.6760,\,1.6843,\,1.6926\}\\
\end{aligned}
\end{equation}
and
\begin{equation}
    \tilde{Q}= I_2 \times 10^{-4}, \qquad \tilde{R}= I_2 \times 10^{-2}, \qquad \tilde{S} = 0.
\end{equation}
These candidate models are stabilized by designing a controller for the nominal system with observer gain $K^{[0]}$ the Kalman gain and with feedback gain $F^{[0]}$ such that $\tilde{A}^{[0]} - \tilde{B}^{[0]} F^{[0]}$ has eigenvalues at $0.94$ and $0.95$. The feedforward gain $G^{[0]}$ is chosen such that the input $u_k$ to the feedback controlled system acts as a reference for the system output.

In this proof-of-concept demonstration we regard a quadratic constraint set
\begin{equation}
    \mathcal{U} = \left\{\mathbf{u}_k \;\Big|\; \norm{u_\ell - r_\ell}_2^2 \leq \varepsilon , \quad k+1 \leq \ell \leq k+N \right\}.
\end{equation}
such that the input $u_k$ is bounded around a given reference signal $r_k = [3, 5]^\top$ $\forall k$ with $\varepsilon = 2.5\times10^{-3}$. With the straightforward translation $\mathbf{u}_k' = \mathbf{u}_k - \mathbf{r}_k$ this constraint set is equivalent to \eqref{energyconstr} such that the methodology presented in Sect. \ref{methodology} is still applicable to this case. The closed-loop approach BC is implemented with a horizon of $N=5$.

The results are presented in Fig. \ref{realization_FB}. According to the online check presented in Sect. \ref{onlinecheck}, all computations were done in the concave domain of the Bhattacharyya coefficients. From Fig. \ref{realization_FB} it can be concluded that it is impossible to reliably diagnose the model without using an auxiliary input. With the closed-loop approach BC however, an input is computed with a user-defined energy limit so that the system is minimally disturbed while still being able to diagnose the correct model after 400 time steps.

\begin{figure}[bh]
\centering
\includegraphics[width=0.75\columnwidth]{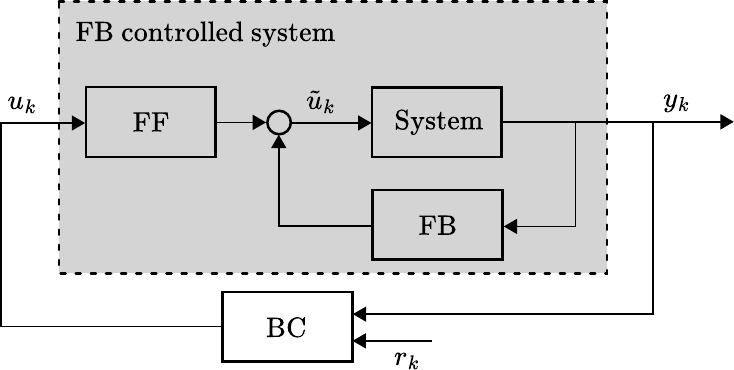}
\caption{Online input design for discrimination of models for a feedback (FB) controlled system with observer, feedforward (FF) and feedback gains $K^{[0]}$, $G^{[0]}$ and $F^{[0]}$, respectively. The online input design method BC can also be replaced with one of the other methods proposed in this manuscript $\Sigma$BC or QTA.}
\label{scheme}
\end{figure}

\begin{figure}[bh]
\centering
\includegraphics[width=0.99\columnwidth]{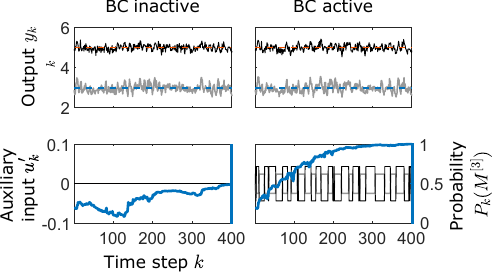}
\caption{A feedback controlled system with and without auxiliary input generated by BC (right and left, respectively). The system outputs $y_k$ (top) approximately track a reference $r_k = [3,5]^\top$ (blue and red dashed, respectively) even though the auxiliary input $u_k'=u_k-r_k$ is applied. However, the correct model could be diagnosed with probability $P_{400}(M^{[3]}) \approx 1$ using BC, whereas this is not possible with $u_k'=0$ $\forall k$.}
\label{realization_FB}
\end{figure}

\section{Conclusion}\label{conclusion}
This manuscript improves the trade-off between the three factors high accuracy, small number of consecutive measurements and low computational effort, with respect to existing stochastic CLAFD methods such as \cite{Paulson2018}. The proposed approach considers the weighted sum of Bhattacharyya coefficients as bound on the error probability, which is shown to be concave in the case of close linear candidate models with significant Gaussian process and measurement noise. In addition, a quadratic Taylor approximation of the error bound is proposed for a further speedup of the computations without significantly sacrificing the discrimination performance.
Simulation results show that the newly proposed approaches require a smaller number of measurements than the open-loop approach, whereas the approach in \cite{Paulson2018} frequently fails to decide at all with the predefined confidence. By additionally considering a similar computational effort, it can be concluded that the newly proposed approaches have favorable trade-off with respect to the state-of-the-art closed-loop method in \cite{Paulson2018}.

In order to further improve the online input design, future research is suggested to focus on minimizing the actual error probability instead of only the bound. The approach proposed in this manuscript can be used to initialize this nonlinear optimization problem.

\appendix

\section{Applicability for feedback controlled systems} \label{FBC}
Consider the open-loop candidate models
\begin{equation}
M^{[i]}: \left\{ \begin{array}{rl} x_{k+1} &= \tilde{A}^{[i]} x_{k} + \tilde{B}^{[i]} \tilde{u}_k + \tilde{w}_k\\
y_k &= \tilde{C}^{[i]} x_{k}  +  \tilde{v}_k \end{array}\right.
\label{OLmodels}
\end{equation}
with the joint covariance matrix of the noise sequences
\begin{equation}
    \mathbb{E}\left[\begin{bmatrix}
        \tilde{v}_k \\ \tilde{w}_k
    \end{bmatrix} \begin{bmatrix}
        \tilde{v}_\ell^\top & \tilde{w}_\ell^\top
    \end{bmatrix}\right] = 
    \left\{ \begin{array}{ll} \begin{bmatrix}
        \tilde{R} & \tilde{S}^\top \\ \tilde{S} & \tilde{Q}
    \end{bmatrix} & \text{if } k = \ell, \\
    0 & \text{otherwise.}
    \end{array} \right.
    \label{OLmodelsnoise}
\end{equation}
A feedback controller is constructed for the nominal model 
\begin{equation}
    \tilde{u}_k = -F^{[0]} \hat{x}_{k|k-1}^{[0]} + G^{[0]} u_{k}
\end{equation}
with $F^{[0]}$ and $G^{[0]}$ the feedback and feedforward gains. The state estimate is given by
\begin{equation}
    \hat{x}_{k+1|k}^{[0]} = \tilde{A}^{[0]}\hat{x}_{k|k-1}^{[0]} + \tilde{B}^{[0]} \tilde{u}(k) + K^{[0]} \Big(y(k) - \tilde{C}^{[0]} \hat{x}_{k|k-1}^{[0]}\Big)
\end{equation}
with $K^{[0]}$ the observer gain.
The resulting closed-loop dynamics of the candidate models are described by \eqref{models} and \eqref{modelsnoise} with
\begin{equation}
    \begin{aligned}
        A^{[i]} &= \begin{bmatrix}
            \tilde{A}^{[i]} & -\tilde{B}^{[i]} F^{[0]} \\
        K^{[0]} \tilde{C}^{[i]} & \tilde{A}^{[0]} - \tilde{B}^{[0]} F^{[0]} - K^{[0]} \tilde{C}^{[0]} 
    \end{bmatrix},\\
        B^{[i]} &= \begin{bmatrix}
            \tilde{B}^{[i]} \\
            \tilde{B}^{[0]}
        \end{bmatrix} G^{[0]}, \qquad C^{[i]} = \begin{bmatrix}
            \tilde{C}^{[i]} & 0
        \end{bmatrix}, \\
        Q &= \begin{bmatrix}
            \tilde{Q} & \tilde{S}(K^{[0]})^\top \\
            K^{[0]} \tilde{S}^\top & K^{[0]} \tilde{R} (K^{[0]})^\top
        \end{bmatrix}, \\
        S &= \begin{bmatrix}
            \tilde{S} \\ K^{[0]} \tilde{R}
        \end{bmatrix}, \qquad\qquad R = \tilde{R}.
    \end{aligned}
\end{equation}

\section{Definitions for \eqref{bdist}} \label{definitions}
Whereas the definitions of $H^{[ij]}$ and $c^{[ij]}$ are given in \eqref{Hc}, the variable $h^{[ij]}$ is defined as
\begin{equation}
    \begin{aligned}
    h^{[ij]} &= \frac{1}{4} (\zeta^{[ij]})^\top (\Omega^{[ij]})^{-1} \zeta^{[ij]}\\
    &\qquad+ \frac{1}{2} \log\left(\frac{\left|\frac{1}{2} \Omega^{[ij]} \right|}{\sqrt{\left|\mathbf{\Sigma}_{k|k}^{[i]}\right|\left|\mathbf{\Sigma}_{k|k}^{[j]}\right|}} \right)
    \end{aligned}
    \label{defh}
\end{equation}
with $|\cdot|$ indicating the determinant operator.

The boldface matrices and Toeplitz matrix in \eqref{OmGamZet} are given by the state-space matrices in \eqref{models} as (with indices $i$ omitted for clarity)
\begin{equation}
    \begin{aligned}
    \mathbf{A} &= \begin{bmatrix} I\\A\\\vdots\\A^{N-1}    \end{bmatrix}, \qquad \mathbf{B} = I_N \otimes B,\qquad \mathbf{C} = I_N \otimes C,\\
    \mathcal{T}_{A} &= \begin{bmatrix} 
    0 & 0 & \hdots & 0 & 0\\
    I & 0 & \hdots & 0 & 0
    \\A & I & \ddots & \vdots & \vdots \\
    \vdots &\ddots & \ddots & 0 & 0
    \\A^{N-2} & \hdots & A & I & 0  \end{bmatrix}
    \end{aligned}
    \nonumber
\end{equation}
with $I_N$ the $N$-by-$N$ identity matrix and $\otimes$ the Kronecker product. 
Moreover, the equations for models \eqref{models} can be expanded to
\begin{equation}
    \begin{aligned}
    \mathbf{x}_{k} &= \mathbf{A} x_{k+1} + \mathcal{T}_{A}\mathbf{B} \mathbf{u}_k + \mathcal{T}_{A}\mathbf{w}_k\\
    \mathbf{y}_{k} &= \mathbf{C} \mathbf{x}_{k} + \mathbf{v}_k
    \end{aligned}
\end{equation}
with $\mathbf{w}_k \sim \mathcal{N}(0,\mathbf{Q})$ and $\mathbf{v}_k \sim \mathcal{N}(0,\mathbf{R})$.
Given the joint covariance matrix in \eqref{modelsnoise}, the extended covariance matrices are defined as
\begin{equation}
\begin{aligned}
    \mathbf{Q} &= I_N \otimes Q, \quad \mathbf{R} = I_N \otimes R, \\
    \mathbf{S} &= \mathbb{E}\left[\mathbf{w}_k \mathbf{v}_{k+1}^\top\right] = \begin{bmatrix}
        0 & 0 & \hdots & 0\\
        S & 0 & \hdots & 0\\
        \vdots & \ddots & \ddots & \vdots \\
        0 & \hdots & S & 0
    \end{bmatrix}.
\end{aligned}
\end{equation}

The predicted output sequence with corresponding covariance matrix is
\begin{equation}
    \begin{aligned}
    \hat{\mathbf{y}}_{k|k} &= \mathbf{C} \mathbf{A} \hat{x}_{k+1|k} + \mathbf{C}\mathcal{T}_{A} \mathbf{B} \mathbf{u}_k\\
    \mathbf{\Sigma}_{k|k} &= \mathbf{C} \mathbf{A} \Xi_{k+1|k} \mathbf{A}^\top \mathbf{C}^\top + \mathbf{C} \mathcal{T}_{A} \mathbf{Q} \mathcal{T}_{A}^\top \mathbf{C}^\top \\
    & \quad + \mathbf{C} \mathcal{T}_A \mathbf{S} + \mathbf{S}^\top \mathcal{T}_A^\top \mathbf{C}^\top  + \mathbf{R} 
    \end{aligned}
\end{equation}
with $\Xi_{k+1|k}$ the state-error covariance matrix corresponding to state estimate $\hat{x}_{k+1|k}$. The covariance matrix $\mathbf{\Sigma}_{k|k}$ is then used in \eqref{OmGamZet} and \eqref{defh} to determine the Bhattacharyya distance in \eqref{bdist}.

\section{Solution of \eqref{min}} \label{concavitysol}
With the definitions in \eqref{defs}, optimization \eqref{min} can be rewritten as
\begin{equation}
\begin{aligned}
g^* = \arg \min_g \; &g^\top \Lambda_1^{-1} g - c^\top U_1 \Lambda_1^{-\frac{3}{2}} g\\
\text{s.t. } &g^\top g = \frac{1}{2},
\end{aligned}
\end{equation}
which in turn is equivalent to
\begin{equation}
\begin{aligned}
q^* = \arg \min_q \; &q^\top \Lambda_1^{-1} q + b^\top q\\
\text{s.t. } &q^\top q = \frac{1}{2},
\end{aligned}
\label{eqprob}
\end{equation}
with $\Lambda_1$ a diagonal matrix with elements $\lambda_1 \geq \lambda_2 \geq \hdots \geq \lambda_n 
> 0$ and $b \geq 0$. The Lagrangian of \eqref{eqprob} equals
\begin{equation}
\mathcal{L}(q,\tau) = q^\top \Lambda_1^{-1} q + b^\top q - \tau \left(q^\top q - \frac{1}{2}\right).
\end{equation}
The solution is given by the stationary point of the Lagrangian and should 
therefore satisfy
\begin{equation}\label{eq: Lagrangian}
	\left\{
	\begin{aligned}
	q_1 &= \frac{b_1}{2 \left(\tau - \lambda_1^{-1} \right)}\\
	&\vdots\\
	q_n &= \frac{b_n}{2 \left(\tau - \lambda_n^{-1} \right)}\\
	q^\top q &= \frac{1}{2}
	\end{aligned}
\right. .
\end{equation}
with $b_i,q_i$ the $i$th element in $b$ or $q$, respectively. This system of equation can be seen as intersection 
of a parametrically defined critical curve $q(\tau)$ defined by its first 
$n$ equations with the sphere $\mathcal{B}$ of 
radius $1/\sqrt{2}$.
On the other hand, it can be rewritten into one polynomial equation in 
variable $\tau$ of degree $2n$, and there exist no more than $2n$ 
intersection points corresponding to the polynomial roots, which can be found 
numerically by  root-finding methods. 
For each of the intersection points, the objective function of \eqref{eqprob} 
can be calculated and the point giving the minimal value provides the 
solution  of \eqref{eqprob}.
The searching range of $\tau$ for the root-finding methods can be reduced 
significantly by using the following considerations. 

Due to the symmetry of the constraints of \eqref{eqprob} and $b$ being 
non-negative, the solution of \eqref{eqprob} lies in the negative orthant, and 
therefore $\tau \in (- \infty, \lambda_1^{-1} ).$
On this interval, each of the first $n$ equations of \eqref{eq: Lagrangian} is 
negative, continuous, and monotonically decreasing, and thus the point 
$q(\tau)$ on the 
critical curve  goes out of the origin (at $\tau = -\infty)$ with 
continuously 
and monotonically increasing distance to the origin, with $|q(\lambda_1^{-1})| = 
\infty.$
Therefore, the range of $\tau$ can be reduced to that part of the curve 
which lies outside the inscribed hypercube in $\mathcal{B}$ and inside the 
circumscribed hypercube around $\mathcal{B}$.
The inscribed hypercube is given by conditions $\{|q_i|\leq  1/\sqrt{2n}, \  
i=1,\ldots n\}$, and thus for $$\tau \leq \tau_- 
\stackrel{\text{def}}{=}\min_i \lambda_i^{-1} - \sqrt{\frac{n}{2}} b_i  
,$$ the critical curve is still inside the inscribed hypercube.
Similarly, the circumscribed hypercube is given by  $\{|q_i| \leq 1/\sqrt{2}, 
\  i=1,\ldots n\}$ and the curve is still inside it for $$\tau \leq 
\tau_+ 
\stackrel{\text{def}}{=}\min_i \lambda_i^{-1} - \frac{1}{\sqrt{2}} b_i. $$
The root-searching method can be thus reduced to the interval $[\tau_-, 
\tau_+]$.
Finally, because $|q(\tau)|$ is continuous and monotone on this interval, 
the solution exists and is unique.

\end{document}